\documentclass[journal,draftcls,onecolumn,12pt,twoside]{IEEEtranTCOM}
\normalsize

\ifCLASSINFOpdf
\else
\fi

\usepackage{graphicx}
\usepackage{array}
\usepackage{amsmath}
\usepackage{float}
\usepackage{float}
\usepackage{multirow}
\usepackage{array}
\usepackage{cancel}
\usepackage[normalem]{ulem}
\usepackage{tabu}
\usepackage{amsthm}
\usepackage{epsfig}
\usepackage{epstopdf}
\usepackage{epsf}
\usepackage{color}
\usepackage[english]{babel}
\newtheorem{exmp}{Example}[section]
\theoremstyle{definition}
\newtheorem{theorem}{Theorem}
\newtheorem{lemma}{Lemma}

\newtheorem{corollary}{Corollary}

\usepackage{amssymb}
\usepackage{lipsum}
\usepackage{mathtools}
\usepackage{cuted}

\DeclareMathOperator*{\argmax}{arg\,max}

\begin{document}
%
\title{On the Uniqueness of Binary Quantizers for Maximizing Mutual Information}
%
%
%

\author{Thuan Nguyen and Thinh Nguyen, \textit{Senior Member, IEEE}
\thanks{Thuan Nguyen is with the School of  Electrical Engineering and Computer Science, Oregon State University, Oregon, OR,
97331 USA, e-mail: (nguyeth9@oregonstate.edu).}
\thanks{Thinh Nguyen is with the School of  Electrical Engineering and Computer Science, Oregon State University, Oregon, OR,
97331 USA, e-mail: (thinhq@eecs.oregonstate.edu).}}

%
%

\markboth{IEEE Transactions on Communications}%
{Submitted paper}
%



\maketitle

\begin{abstract}
We consider a channel with a binary input $X$ being corrupted by a continuous-valued noise that results in a continuous-valued output $Y$.  An optimal binary quantizer is used to quantize the continuous-valued output $Y$ to the final binary output $Z$ to maximize the mutual information $I(X; Z)$. We show that when the ratio of the channel conditional density $r(y) = \frac{P(Y=y|X=0)}{P(Y = y|X=1)}$ is a strictly increasing/decreasing function of $y$, then a quantizer having a single threshold can maximize mutual information.   
Furthermore, we show that an optimal quantizer (possibly with multiple thresholds) is the one with the thresholding vector whose elements are all the solutions of $r(y)=r^*$ for some constant $r^*>0$. Interestingly, the optimal constant $r^*$ is unique.  This uniqueness property allows for fast algorithmic implementation such as a bisection algorithm to find the optimal quantizer.  Our results also confirm some previous results using alternative elementary proofs.  We show some  numerical examples of applying our results to channels with additive Gaussian noises.

\end{abstract}


%
\IEEEpeerreviewmaketitle

\section{Introduction}

Quantization techniques play a vital role in signal processing, communication, and information theory. A classical quantization technique maps  a given real number to an element in a given finite discrete set that minimizes/maximizes a certain objective.  Vector quantization (VC) extends the classical quantization to allow the input to take on the real-valued vector \cite{lloyd1982least, gersho2012vector}. In compression, quantization is often used to minimize the distortion (e.g. mean square error (MSE)) between the original data and its quantized version \cite{goldberg1986image, gong2014compressing}.  In graphics, color quantization is used to reduce the number of colors in the images for displays with various capabilities \cite{akarun1997adaptive}.  In communication, quantization is often used to minimize the decoding errors.  Broadly, any conversion of a high-resolution signal to a low-resolution signal requires quantization. In this paper, we consider the quantization in the context of a communication channel where the transmitted binary signal is corrupted by a continuous noise, resulting in a continuous-valued signal at the receiver.  To recover the transmitted signal, the receiver performs a quantization algorithm that maps the received continuous-valued signal to the quantized signal such that the objective function between the input and the quantized output is maximized/minimized.   There is a rich literature on quantizer design that minimizes various objectives. One popular objective is to minimize the average decoding error.  Another fundamental objective is to maximize the mutual information between the discrete transmitted inputs and the quantized outputs. Equivalently, this objective  minimizes the information loss between the inputs and the outputs, and is related to the capacity of the channel. Specifically, for a given discrete memoryless channel (DMC) specified by a channel matrix $M$, its capacity is found by maximizing the mutual information between the input and the output with respect to the input distribution $p$ \cite{cover2012elements}, \cite{nguyen2018closed}.  On the other hand,  our work is focused on maximizing the mutual information with respect to the quantization parameters, i.e, it is equivalent to designing  a channel matrix $M$  for a fixed distribution $p$ that maximizes the capacity.  This situation often arises in real-world scenarios where the distribution of input is already given. In addition, many recent works have proposed to use quantization strategies that maximize the mutual information in the designs of low density parity check codes (LPDC) \cite{romero2015decoding, wang2011soft} and polar codes \cite{tal2011construct}.  

We consider a channel with binary input $X$ that is corrupted by a given continuous noise to produce continuous-valued output $Y$.  An optimal binary quantizer is then used to quantize the continuous-valued output $Y$ to the final binary output $Z$ to maximize the mutual information $I(X;Z)$.  We show that when the ratio of the channel conditional density $r(y) = \frac{P(Y=y|X=0)}{P(Y = y|X=1)}$ is a strictly increasing/decreasing function of $y$, then a quantizer having a single threshold can maximize mutual information.   
	Furthermore, we show that an optimal quantizer (possibly with multiple thresholds) is the one with the thresholding vector whose elements are all the solutions of $r(y)=r^*$ for some constant $r^*>0$. Interestingly, the optimal constant $r^*$ is unique.  This uniqueness property allows for fast algorithmic implementation such as a bisection algorithm to find the optimal quantizer.  Our results also confirm some previous results using alternative elementary proofs.

The outline of the paper is as follows.  First, we discuss a few related works in Section \ref{sec:related work}.  In Section \ref{sec:problem description}, we formulate the problem of designing the optimal quantizer that maximizes the mutual information.  In Section \ref{sec:structure}, we describe the structure of optimal quantizers.  In Section \ref{sec:unique}, we show the existence of a unique optimal quantizer  for a given number of thresholds, and provide an efficient algorithm for finding it.  Finally, we provide examples and numerical results to verify our contributions in Section \ref{sec:simulations}.

\section{Related Work}
\label{sec:related work}

Research on quantization techniques has a long history, including many earliest works in 1960s \cite{max1960quantizing} that aim to minimize the distortion between the original signal and the quantized signal.  
From a communication perspective, designing the quantizers that maximize the information capacity for Gaussian channels have also been proposed in 1970s \cite{smith1971information}. Recently, in constructing efficient codes such as LDPC and polar codes, a number of works have made use of quantizers that maximize the mutual information  \cite{romero2015decoding, wang2011soft, tal2011construct}.  Many advanced quantization algorithms have also been proposed to maximize the mutual information between the input and the quantized output over the past decade \cite{kurkoski2014quantization}, \cite{mathar2013threshold}, \cite{sakai2014suboptimal}, \cite{iwata2014quantizer}, \cite{winkelbauer2013channel}, \cite{koch2013low}. In \cite{kurkoski2014quantization}, the channel is assumed to have discrete input and discrete output,  and the optimal quantizers can be found efficiently using dynamic programming that has polynomial time complexity  \cite{DBLP:journals/corr/abs-1901-01659}.  On the other hand, we study the channels with discrete binary inputs and continuous-valued outputs which are then quantized to binary outputs. The continuous-valued output is a direct result of the channel conditional density.  We note that it is possible to first discretize the continuous-valued output, then use the existing quantization algorithms for the discrete input-discrete output channels \cite{kurkoski2014quantization}.  However, in many scenarios, this may result in loss of efficiencies.  In particular, many analytical and computational techniques for dealing with continuous-valued functions are more efficient than their discrete counterparts. 

Our work is also related to the classification problem in learning theory. Brushtein et al. gave the condition on the existence of an optimal quantizer which minimizes the impurity of partitions \cite{burshtein1992minimum}. Because of the similarity between maximizing mutual information and minimizing conditional entropy function \cite{kurkoski2014quantization}, \cite{kurkoski2017single}, the result in \cite{burshtein1992minimum} can be applied for finding the optimal quantizer.  A similar result also can be found in \cite{coppersmith1999partitioning}. In \cite{zhang2016low},  Zhang et al. show that finding an optimal quantizer is equivalent to finding an optimal clustering. Therefore,  a locally optimal solution can be found using k-means algorithm with the Kullback-Leibler divergence as the distance metric. Recently, there have also been many works on approximating the optimal clustering that minimize the impurity function for high dimensional data \cite{nazer2017information}, \cite{laber2018binary}, \cite{cicalese2019new}.

There are also works on designing quantizers that maximize the channel capacity by maximizing the mutual information over both quantization parameters and the input probability mass function (pmf).  This problem remains to be a hard problem \cite{mathar2013threshold},  \cite{kurkoski2017single}, \cite{nguyen2018capacities}, \cite{alirezaei2015optimum}, \cite{singh2009limits}, \cite{mumey2003optimal}.  Although the mutual information is a convex function in the input pmf, it is not a convex function in the quantization parameters.  As such, many successful convex optimization techniques for finding the optimal solution are not applicable. In \cite{nguyen2018capacities}, a heuristic near optimal quantization algorithm is proposed.  However, the algorithm only works well when the SNR ratio is high. In \cite{mathar2013threshold}, R. Mathar et al. investigate an optimal quantization strategy for binary input-multiple output channels using two support points. These results are only applicable to approximate the optimal point between two supporting points. On the other hand, Kurkoski et al. solve the optimal quantization for a discrete input-discrete output binary channel using a backward channel \cite{kurkoski2017single}.  This technique can find the optimal quantizer with the complexity $O(N)$ using an exhaustive search, where $N$ is the number of discretized levels. This work also shows that there exists an optimal quantizer using only a single threshold if the conditional noise density satisfies the log-likelihood condition.  On the other hand, our work describes the generalized conditions for the existence of a single threshold optimal quantizer together with the uniqueness property that allows for designing fast algorithms to find an optimal quantizer.

\section{Problem description}
\label{sec:problem description}
We consider the channel shown in Fig. \ref{fig:setup} where the binary signals $ x \in X=\{0,1\}$ are transmitted and corrupted by a continuous noise source to produce a continuous-valued output $y \in \mathbb{R}$ at the receiver. Specifically, $y$ is specified by the a channel conditional density $p(y|x)$.  $p(y|x)$ models the distortion caused by noise. The receiver recovers the original binary signal $x$ by decoding the received continuous-valued signal $y$ to $z \in Z=\{0,1\}$ using a quantizer $Q$. Since $y\in \mathbb{R}$, the quantization parameters can be specified by a thresholding vector 
$$\textbf{h}= (h_1, h_2, \dots,h_n) \in \mathbb{R}^n,$$
with $h_1 < h_2 < \dots < h_{n-1} < h_n$, where $n$ is assumed a finite number.  
Theoretically, it might be perceivably possible to construct the conditional densities $p(y|x_0)$ and $p(y|x_1)$ such that the optimal quantizer might consist an infinite number of thresholds.  On the other hand, for a practical implementation, especially when the quantizer is implemented using a lookup table, then a finite number of thresholds must be used.  To that end, the optimal quantizer in this paper refers to the best quantizer in the class of all quantizers with a finite number of thresholds.

\begin{figure}
	\centering
	\includegraphics[width=0.7\linewidth]{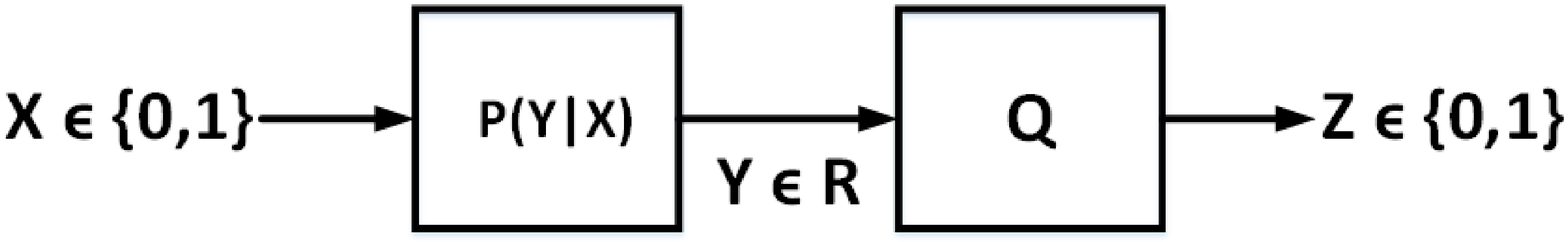}
	\caption{Channel model: binary input $X$ is corrupted by continuous noise to result in continuous-valued $Y$ at the receiver. The receiver attempts to recover $X$ by quantizing $Y$ into binary signal $Z$.}
	\label{fig:setup}
\end{figure}

In particular, $\textbf{h}$ induces $n+1$ disjoint partitions: $$H_1 = (-\infty, h_1), H_2 = [h_1, h_2), \dots, H_n = [h_{n-1}, h_n), H_{n+1} = [h_n,\infty).$$ 
Let $\mathbb{H} = \bigcup_{i \in odd} H_i$ and  $\bar{\mathbb{H}} = \bigcup_{i \in even} H_i$, then $\mathbb{H} \cap \bar{\mathbb{H}} =\emptyset$ and $\mathbb{H} \cup \bar{\mathbb{H}} =\mathbb{R}$. 

The receiver uses a quantizer $Q:Y \rightarrow Z$ to quantize $Y$ to $Z$ as:
\begin{equation}
 \label{eq:decoding}
Z = \begin{cases} 
	0  & \text{if  } Y \in \mathbb{H}, \\
	1  & \text{if  }  Y \in \bar{\mathbb{H}}.
\end{cases}
\end{equation}

Note that we can also switch the rule such that $Q$ quantizes $Y$ to $Z=1$ if $y \in \mathbb{H}$ and quantizes $Y$ to $Z=0$ if $y \in \bar{\mathbb{H}}$.   The main point is that $\textbf{h}$ divides $\mathbb{R}$ into $n+1$ contiguous disjoint segments, each maps to either 0 or 1 alternatively.
Our goal is to design an optimal quantizer $Q^*$, specifically $\textbf{h}^*$ that maximizes the mutual information $I(X;Z)$ between the input $X$ and the quantized output $Z$:
\begin{equation}
\label{eq:maximization}
\textbf{h}^*= \argmax_{\bf{h}} I(X;Z).
\end{equation} 
We note that both the values of thresholds $h_i$'s and the number of thresholds $n$ are the optimization variables.
The maximization in (\ref{eq:maximization}) assumes that the input probability mass function $p(x)$ and the channel conditional density $p(y|x)$ are given. 

\section{Optimal Quantizer Structure}
\label{sec:structure}

For convenience, we use the following notations:

\begin{enumerate}
	\item $\textbf{p} = (p_0,  p_1)$  denotes the probability mass function for the input $X$, with $p_0 =  P(X=0)$ and $p_1 = P(X=1)$.
	\item $\textbf{q} = (q_0, q_1)$ denotes probability mass function for the output $Z$, with  $q_0 =  P(Z=0)$ and $q_1 = P(Z=1)$.
	\item  $\phi_0(y)=p(y|x=0)$ and $\phi_1(y)=p(y|x=1)$ denote conditional density functions of the received signal $Y$ given the input signal $X = 0$ and $X=1$, respectively.
\end{enumerate}

Furthermore, we make two following assumptions:\\
{\bf Assumptions:}	
\begin{enumerate}
	\item  $r(y)=\dfrac{\phi_0(y)}{\phi_1(y)}$ will play a central role this paper.  All the results in this paper assume that $r(y)$ is a continuous function, and has a finite number of stationary points. Equivalently, $r(y) = r'$ has a finite number of solutions for any constant $r' > 0$.  Note that this assumption will hold for most $\phi_0(y)$ and $\phi_1(y)$.	
	\item Both $\phi_0(y)$ and $\phi_1(y)$ are differentiable everywhere.
 \end{enumerate}

Using the notations and the assumptions above, a 2$\times2$ channel matrix $A$ associated with a discrete memoryless channel (DMC) with input $X$ and output $Z$ is:
\[
A=
\begin{bmatrix}
A_{11} & 1-A_{11}\\
1-A_{22} & A_{22}
\end{bmatrix},
\]

where
\begin{equation}
\label{eq: definition a11}
A_{11} =\int_{y \in \mathbb{H}}^{}\phi_0(y)dy,
\end{equation}
\begin{equation}
\label{eq: definition a22}
A_{22} =\int_{y \in \bar{\mathbb{H}}}^{}\phi_1(y)dy.
\end{equation}

The simplest quantizer (decoding scheme) uses only a single threshold to quantize a continuous received signal into binary outputs.  Specifically, 

$$Z = 
\begin{cases}
 	 0 & \text{if } Y < h_1, \\
	 1 & \text{otherwise.}
\end{cases}
$$
In general, this quantizer is not optimal, i.e., does not maximize the mutual information $I(X;Z)$.  Using the results of Burnstein et al.  \cite{burshtein1992minimum},   Kurkoski et al. \cite{kurkoski2017single} showed a sufficient condition on $p(y|x)$ for which the single threshold quantizer is indeed an optimal quantizer.  Our first contribution is to show that the optimal binary quantizer with multiple thresholds, specified by a thresholding vector $\textbf{h}^*=(h_1^*,h_2^*,\dots,h_n^*)$ with $h_i^* < h_{i+1}^*$, must satisfy the conditions stated in the Theorem \ref{theorem: 1}.

\begin{theorem}
\label{theorem: 1}
Let $\textbf{h}^*=(h_1^*,\dots,h_n^*)$ be a thresholding vector of an optimal quantizer $Q^*$, then:
\begin{equation}
\label{eq: relation threshold}
\dfrac{\phi_0(h_i^*)}{\phi_1(h_i^*)} = \dfrac{\phi_0(h_j^*)}{\phi_1(h_j^*)} = r^*,
\end{equation}
for $\forall$ $i, j \in \{1,2,\dots,n\}$ and some optimal constant $r^* > 0$. 
\end{theorem}
\begin{proof}
We note that using the optimal thresholding vector $\textbf{h}^*$, the quantization mapping follows 
(\ref{eq:decoding}).  $\textbf{h}^*$ divides $\mathbb{R}$ into $n+1$ contiguous disjoint segments, each maps to either 0 or 1 alternatively.
The overall DMC in Fig. \ref{fig:setup} has the channel matrix 
\[
A^*=
  \begin{bmatrix}
A_{11} & A_{12}\\
A_{21} & A_{22}
  \end{bmatrix},
\]
and the mutual information can be written as a function of $\textbf{h}$ as: 
\begin{equation}
I(\textbf{h})=H(Z)-H(Z|X)=H(q_0)-[p_0H(A_{11})+p_1H(A_{22})],
\end{equation}
where for any $w \in [0,1]$, $H(w)=-[w\log(w)+(1-w)\log(1-w)]$ and $q_0=P(Z=0)=p_0A_{11}+p_1A_{21}$.

This is an optimization problem that maximizes $I(\textbf{h})$.  The theory of optimization requires that an optimal point must satisfy the KKT conditions \cite{boyd2004convex}.  In particular, define the Lagrangian function as:
\begin{equation}
\label{eq:lagrangian}
L(\textbf{h},\lambda) = I(\textbf{h}) + \sum_{i=1}^{n-1}{\lambda_i (h_i - h_{i+1})},
\end{equation}
then the KKT conditions \cite{boyd2004convex} states that, an optimal point $\textbf{h}^*$ and $\lambda^*= (\lambda^*_1, \lambda^*_2, \dots, \lambda^*_{n-1})$ must satisfy:

\begin{equation}
\label{eq:kkt1}
\begin{cases}
\frac{\partial{L(\textbf{h}, \lambda}}{\partial{h_i}}|_{\textbf{h} = \textbf{h}^*,\lambda = \lambda^*}, i = 1, 2, \dots, n-1,\\
\lambda_i^*(h_i - h_{i+1}) = 0, i = 1, 2, \dots, n-1, \\
\lambda_i^* \ge 0, i = 1, 2, \dots, n-1.
\end{cases}
\end{equation}
Since the structure of the quantizer requires that $h_i < h_{i+1}$, the second and the third conditions in (\ref{eq:kkt1}) together imply that $\lambda_i^* = 0, i = 1, 2, \dots, n-1$. Consequently, from (\ref{eq:lagrangian}) and the first condition in (\ref{eq:kkt1}), we have:

$$\frac{\partial{L(\textbf{h}, \lambda)}}{\partial{h_i}}|_{\textbf{h} = \textbf{h}^*,\lambda = \lambda^*} = \frac{\partial{I(\textbf{h})}}{\partial h_i}|_{\textbf{h}= \textbf{h}^*} = 0.$$

The stationary points can be found by setting the partial derivatives with respect to each $h_i$ to zero:
{\small
\begin{eqnarray}
 \frac{\partial I(\textbf{h})}{\partial h_i} &=&(\log\dfrac{1-q_0}{q_0} )\frac{\partial q_0}{\partial h_i} -p_0(\log \dfrac{1-A_{11}}{A_{11}}) \frac{\partial A_{11}}{\partial h_i}-p_1(\log\dfrac{1-A_{22}}{A_{22}})\frac{\partial A_{22}}{\partial h_i} \nonumber\\
&=&(\log\dfrac{1\!-\!q_0}{q_0})(p_0 \frac{\partial A_{11}}{\partial h_i} \!- \!p_1 \frac{\partial A_{22}}{\partial h_i} )\!-\!p_0(\log\dfrac{1-A_{11}}{A_{11}}) \frac{\partial A_{11}}{\partial h_i} \!-\!p_1(\log\dfrac{1-A_{22}}{A_{22}})\frac{\partial A_{22}}{\partial h_i}\label{eq: 2}\\
&=&p_0 \frac{\partial A_{11}}{\partial h_i} (\log \dfrac{1-q_0}{q_0}-\log \dfrac{1-A_{11}}{A_{11}})-p_1 \frac{\partial A_{22}}{\partial h_i} (\log \dfrac{1-q_0}{q_0}+\log \dfrac{1-A_{22}}{A_{22}}) = 0,
\label{eq:derivative}
\end{eqnarray}
}
with (\ref{eq: 2}) due to $q_0=p_0A_{11}+p_1A_{21}=p_0A_{11}+p_1(1-A_{22})$.

Since  $\frac{\partial A_{11}}{\partial h_i}=\phi_0(h_i)$ and $\frac{\partial A_{22}}{\partial h_i}=-\phi_1(h_i)$, from (\ref{eq:derivative}), we have:
\begin{equation}
\label{eq: 3}
\dfrac{\phi_0(h^*_i)}{\phi_1(h^*_i)}=-\dfrac{p_1}{p_0}\dfrac{\log\dfrac{1-q_0}{q_0}+\log\dfrac{1-A_{22}}{A_{22}}}{\log\dfrac{1-q_0}{q_0}-\log\dfrac{1-A_{11}}{A_{11}}}= r^*.
\end{equation}

Since (\ref{eq: 3}) holds for $\forall$ $i$, the RHS of (\ref{eq: 3}) equals to some constant $r^* > 0 $ for a quantizer $Q^*$,  Theorem \ref{theorem: 1} follows.
\end{proof}

\textbf{Remark:} An important  of Theorem \ref{theorem: 1} is as follows.  Suppose the optimal value $r^*$ is given and the equation $r(y) = r^*$ has $m$ solutions: $y_1 < y_2 < \dots <  y_m$. Then,  Theorem \ref{theorem: 1} says that the optimal quantizer must either have its thresholding vector be $(y_1, y_2, \dots, y_m)$ or  one of its ordered subsets, e.g., $(h^*_1, h^*_2) = (y_1, y_3)$, or both.   In Theorem \ref{Theorem: 2} below, we will show that the quantizer whose thresholding vector is all the solutions of $r(y) = r^*$, will be at least as good as any quantizer whose thresholding vector is a ordered subset of the set of all solutions.  Moreover, we will show that $r^*$ is unique, and describe an efficient procedure for finding $r^*$ in Section \ref{sec:unique}.	

\begin{theorem}
\label{Theorem: 2}
Let $y^*_1 < y^*_2< \dots < y^*_n$ be the solutions of $r(y) = r^*$ for the optimal constant $r^* > 0$. Let $Q^n_{r^*}$ be the quantizer whose thresholding vector is all the solutions, i.e., $h^*_i = y^*_i, i = 1, 2, \dots, n$, then for $k < n$, $Q^n_{r^*}$ is at least as good as any quantizer $Q^k_{r^*}$ whose thresholding vector is an ordered subset of the set of $(h^*_1, h^*_2, \dots, h^*_n)$.
\end{theorem}
\begin{proof}
Let $(h^*_1, h^*_2, \dots, h^*_m) $ be an optimal thresholding vector for all the quantizers having $m$ thresholds ($m \le n$).
Let $(z^*_1, z^*_2, \dots, z^*_{m-1})$ be an optimal thresholding vector for all quantizers having $m-1$ thresholds.
The mutual information can be written as a function of these quantizers as:  $I(h^*_1,  h^*_2, \dots, h^*_m)$ and $I(z^*_1, z^*_2, \dots, z^*_{m-1})$.  
We will first show that $I(h^*_1,  h^*_2, \dots, h^*_m) \geq I(z^*_1, z^*_2, \dots, z^*_{m-1})$, for any $m > 0$.
This will be proved using contradiction.  

Assume that $I(h^*_1,  h^*_2, \dots, h^*_m) < I(z^*_1, z^*_2, \dots, z^*_{m-1})$, then
\begin{equation}
\label{eq: theorem2-1}
I(z^*_1, z^*_2, \dots, z^*_{m-1}) = I(h^*_1,  h^*_2, \dots, h^*_m) + \delta,
\end{equation}
where $\delta$ is a positive constant. 

Since $(h^*_1, h^*_2, \dots, h^*_m) $ is optimal,
\begin{equation}
\label{eq: theorem2-2a}
 I(h^*_1,  h^*_2, \dots, h^*_m) \geq  I(h_1,  h_2, \dots, h_{m-1},h_m),
\end{equation}
for any $h_i < h_{i+1}$, $i = 1, 2, \dots, m-1$.

Now replacing $h_i = z^*_i$, for $i = 1, 2, \dots, m-1$ into (\ref{eq: theorem2-2b}), we have:
 
\begin{equation}
\label{eq: theorem2-2b}
I(h^*_1,  h^*_2, \dots, h^*_m) \geq  I(z^*_1,  z^*_2, \dots, z^*_{m-1},h_m).
\end{equation}
Since  $\int^{\infty}_{-\infty}{\phi_i(y)dy} = 1$, $\forall$ $i=1,2$, 
 $$\lim_{y \rightarrow \infty} \phi_i(y)=0, i=1,2.$$  
Consequently,
$$\lim_{h_m \rightarrow \infty}  I(z^*_1,  z^*_2, \dots, z^*_{m-1},h_m)= I(z^*_1,  z^*_2, \dots, z^*_{m-1}).$$
Equivalently,  there exists an $h_m > N_{\epsilon}$ such that
\begin{equation}
\label{eq: theorem2-3}
 | I(z^*_1,  z^*_2, \dots, z^*_{m-1},h_m)- I(z^*_1,  z^*_2, \dots, z^*_{m-1})| \leq \epsilon,
\end{equation}
for any $\epsilon > 0$. Next, we pick a $N_{\epsilon}$ such that $\epsilon < \delta$. Then,
\begin{eqnarray}
 I(h^*_1,  h^*_2, \dots, h^*_m) & = & I(z^*_1,  z^*_2, \dots, z^*_{m-1}) + I(h^*_1,  h^*_2, \dots, h^*_m)  -  I(z^*_1,  z^*_2, \dots, z^*_{m-1}) \nonumber\\
  & \geq &  I(z^*_1,  z^*_2, \dots, z^*_{m-1})- |I(h^*_1,  h^*_2, \dots, h^*_m)  -  I(z^*_1,  z^*_2, \dots, z^*_{m-1})| \nonumber \\
 & \geq &  I(h^*_1,  h^*_2, \dots, h^*_m) +\delta - \epsilon \label{eq: theorem2-5},
\end{eqnarray}
where (\ref{eq: theorem2-5}) is due to (\ref{eq: theorem2-1}) and (\ref{eq: theorem2-3}). Since $\delta - \epsilon > 0$ by assumption, (\ref{eq: theorem2-5}) indicates that  $I(h^*_1,  h^*_2, \dots, h^*_m)$ is strictly greater than itself which is  a contradition. Thus, $I(h^*_1,  h^*_2, \dots, h^*_m) \geq I(z^*_1, z^*_2, \dots, z^*_{m-1})$. 

Next, since $(z^*_1, z^*_2, \dots, z^*_{m-1})$ is an optimal thresholding vector for all quantizers having $m-1$ thresholds, 
$I(z^*_1, z^*_2, \dots, z^*_{m-1}) \geq I(\bar{h}^*_1, \bar{h}^*_2, \dots, \bar{h}^*_{m-1})$ where $(\bar{h}^*_1, \bar{h}^*_2, \dots, \bar{h}^*_{m-1})$ is an arbitrary subset of $(h^*_1, h^*_2, \dots, h^*_m) $. Thus, $I(h^*_1,  h^*_2, \dots, h^*_m) \geq  I(z^*_1, z^*_2, \dots, z^*_{m-1}) \geq I(\bar{h^*_1}, \bar{h^*_2}, \dots, \bar{h^*_{m-1}})$. Consequently, by induction, $Q^n_{r^*}$ is at least as good as any quantizer $Q^k_{r^*}$,  $\forall$ $k < n$.



\end{proof}

\begin{corollary}
\label{cor: 1}
If 
\begin{equation}
\label{eq: brian}
r(y) = \dfrac{\phi_0(y)}{\phi_1(y)} 
\end{equation} is a strictly increasing/decreasing function, then (a) the optimal quantizer consists of only a single threshold $h^*_1$ and (b) it is unique. 
\end{corollary}
\begin{proof}
To prove part (a) we note that since $r(y)$ is a strictly increasing/decreasing function. Therefore,  $r(y_1) \neq r(y_2)$ for $y_1 \neq y_2$.  Thus, (\ref{eq: relation threshold})  will not hold for $h^*_1 \neq h^*_2$.
Consequently, the optimal quantizer has only a single threshold.

The proof of part (b) will be shown in Section \ref{sec:unique}.
\end{proof}

We note that in a previous result \cite{kurkoski2017single},  an optimality condition for a single threshold quantizer is that:
\begin{equation}
\label{eq: brian2}
s(y) = \log{\dfrac{\phi_0(y)}{\phi_1(y)}} 
\end{equation}
is a  monotonic function.  If $\dfrac{\phi_0(y)}{\phi_1(y)}$ is a strictly monotonic function, then previous result is a consequence of Corollary \ref{cor: 1} since  $\log(.)$ is a strictly monotonic function, any strictly monotonic function $\dfrac{\phi_0(y)}{\phi_1(y)}$ results in a strictly monotonic function $s(y)$. 

\begin{corollary}
\label{cor: 2}
If
\begin{equation}
\label{eq: additive}
\phi_0(y-\mu)=\phi_1(y) \text{   for some constant }  \mu,
\end{equation}
and 
$\phi_0(y)$ is a strictly log-concave or log-convex function, then using a single threshold quantizer is optimal.
\end{corollary}
\begin{proof}
Taking derivative of $r(y)$, we have:
\begin{equation}
 \frac{dr(y)}{dy}=\dfrac{\phi_0'(y)\phi_1(y)-\phi_0(y)\phi_1'(y)}{\phi_1(y)^2} >0,
\end{equation}
which is equivalent with:
\begin{equation}
\label{eq: log-concave}
\dfrac{\phi_0'(y)}{\phi_0(y)} >  \dfrac{\phi_1'(y)}{\phi_1(y)}.
\end{equation}
Using (\ref{eq: additive}), we have:
\begin{equation}
\dfrac{\phi_0'(y)}{\phi_0(y)}  >  \dfrac{\phi_0'(y-\mu)}{\phi_0(y-\mu)}.
\end{equation}

Now, a function $f(x)$ is strictly log-convex if and only if  $\dfrac{f'(x)}{f(x)}$ is a strictly increasing  function \cite{boyd2004convex}. Thus, if $\phi_0(y)$ is strictly log-convex, then
\begin{equation}
\label{eq: log-concave inequality}
\dfrac{\phi_0'(y)}{\phi_0(y)} >  \dfrac{\phi_0'(y-\mu)}{\phi_0(y-\mu)}.
\end{equation}
Thus, $r'(y) > 0$ or $r(y)$ is a strictly increasing  function which satisfies the condition for having an optimal single threshold quantizer in Corollary \ref{cor: 1}. A similar proof can be established  for log-concave functions. 
\end{proof}


{\bf Remark:} An important channel model used extensively in communication is the AWGN model with $y_i = x_i + n_i$ where  $n_i$ are independent normal distributions $N(0,\sigma_i)$ with $x_i \in \{x_0, x_1\}$.  Equivalently, $\phi_0(y)$ and $\phi_1(y)$ are the densities of normal distributions with means $x_0$, $x_1$, and variances $\sigma_0$, $\sigma_1$, respectively.  If $\sigma_0=\sigma_1$, a single threshold quantizer is optimal due to Corollary \ref{cor: 2}.  On the other hand, in the general case where $\sigma_0 \neq \sigma_1$, using Theorem \ref{Theorem: 2}, there are two solutions and an optimal quantizer indeed uses two thresholds.  
Examples \ref{ex: 2} and \ref{ex: 3} in Section \ref{sec:simulations} will illustrate this point in detail.

\section{Uniqueness of the Optimal Quantizer for A Given Number of Thresholds}
\label{sec:unique}
In this section, we will show that $r^*$ is unique, i.e., there is a single value of $r^*$ that maximizes the mutual information.  In Theorem \ref{Theorem: 2}, we define $Q^n_{r^*}$ be the quantizer whose thresholding vector is all the $n$ solutions of $r(y) = r^*$.  Combining the result that $r^*$ is unique (to be shown shortly) with Theorem \ref{Theorem: 2}, we can conclude that there is a unique optimal quantizer having $n$ thresholds.  Note that this does not rule out the case that there is another optimal quantizer with $k < n$ thresholds.  However, since $r^*$ is unique, the set of $k$ thresholds must be an ordered subset of the set of the $n$ thresholds.  Furthermore, we describe some efficient methods for finding $r^*$.  Once $r^*$ is determined, we can solve $r(y) = r^*$ to obtain the solutions which are the optimal thresholds.

For ease of analysis, we define a new variable $a$ as:
\begin{equation}
 \label{eq: relate phi0 phi1}
a=\dfrac{p_1\phi_1(y)}{p_0\phi_0(y)+p_1\phi_1(y)}=\dfrac{1}{1+\dfrac{p_0\phi_0(y)}{p_1\phi_1(y)}} =\dfrac{1}{1+\big(\dfrac{p_0}{p_1}\big)r},
\end{equation} 
where $$r = \dfrac{\phi_0(y)}{\phi_1(y)}.$$

We note that $a \in (0,1)$. In addition, the mapping from $r$ to $a$ is a one-to-one mapping. Furthermore, each value of $a$ corresponds to a different value of $r$ which in turn, corresponds to a quantizer in a set of possible quantizers that contains an optimal quantizer.
As an example, Fig. \ref{fig: support} shows the conditional densities $\phi_0(y)$ and $\phi_1(y)$, and the corresponding
$r(y)$ and $u(y)$ are shown in Fig. \ref{fig:cor 3} and Fig. \ref{fig:u}, respectively. Now, the mutual information 
$I(X;Z)$ can be re-written as a function of $a$, and is denoted as $I(X;Z)_{a}$. 

The gist of writing $I(X;Z)$ as a function of $a$ is as follows. While there are many stationary points of $I(X;Z)$ with respect to $\textbf{h}$, i.e., many $\textbf{h}$'s such that $\dfrac{\partial I(X;Z)}{\partial h_i} = 0$, we will show that  $I(X;Z)$  has one stationary point with respect to $a$, i.e., there is only a single $a^*$ for which $\dfrac{dI(X;Z)_a}{da}|_{a = a^*}$ = 0.  The significance of this result is that if $a^*$ is unique, then $r^*$ is unique since $r$ relates to $a$ through a one-to-one mapping. 
Furthermore, the optimal thresholds can be directly determined as the solutions to:
%
\begin{equation}
\label{eq: u-y}
\dfrac{p_1\phi_1(h)}{p_0\phi_0(h)+p_1\phi_1(h)}=a^*.
\end{equation}
To prove this result, we need a number of smaller results stated in the Lemma \ref{lemma: relate first derivation} and 
Lemma \ref{lemma: main result 4} below.  

First, let $$u(y) = \dfrac{p_1\phi_1(y)}{p_0\phi_0(y)+p_1\phi_1(y)}.$$

For given $a$, define $\mathbb{H}_a = \{y: u(y) < a\}$, then
$$ \mathbb{H}_a =  \{ (-\infty, h_1) \ \cup [h_2, h_3) \cup \dots \cup [h_n, +\infty) \}.$$

Similarly,  let $\bar{\mathbb{H}}_a = \{y: u(y) \ge a\}$, then
$$\bar{\mathbb{H}}_a = \mathbb{R} \setminus \mathbb{H}_a = \{ [h_1, h_2 ) \cup [h_3, h_4) \cup \dots \cup [h_{n-1}, h_n) \}.$$

The sets $\mathbb{H}_a$ and $\bar{\mathbb{H}}_a$ together specify a binary quantizer that maps $y$ to $z \in \{0,1\}$, depending on whether $y$ belongs to $\mathbb{H}_a$ or $\bar{\mathbb{H}}_a$ as shown in Fig. \ref{fig:u}.

\begin{figure}
	\centering
	\includegraphics[width=0.5\linewidth]{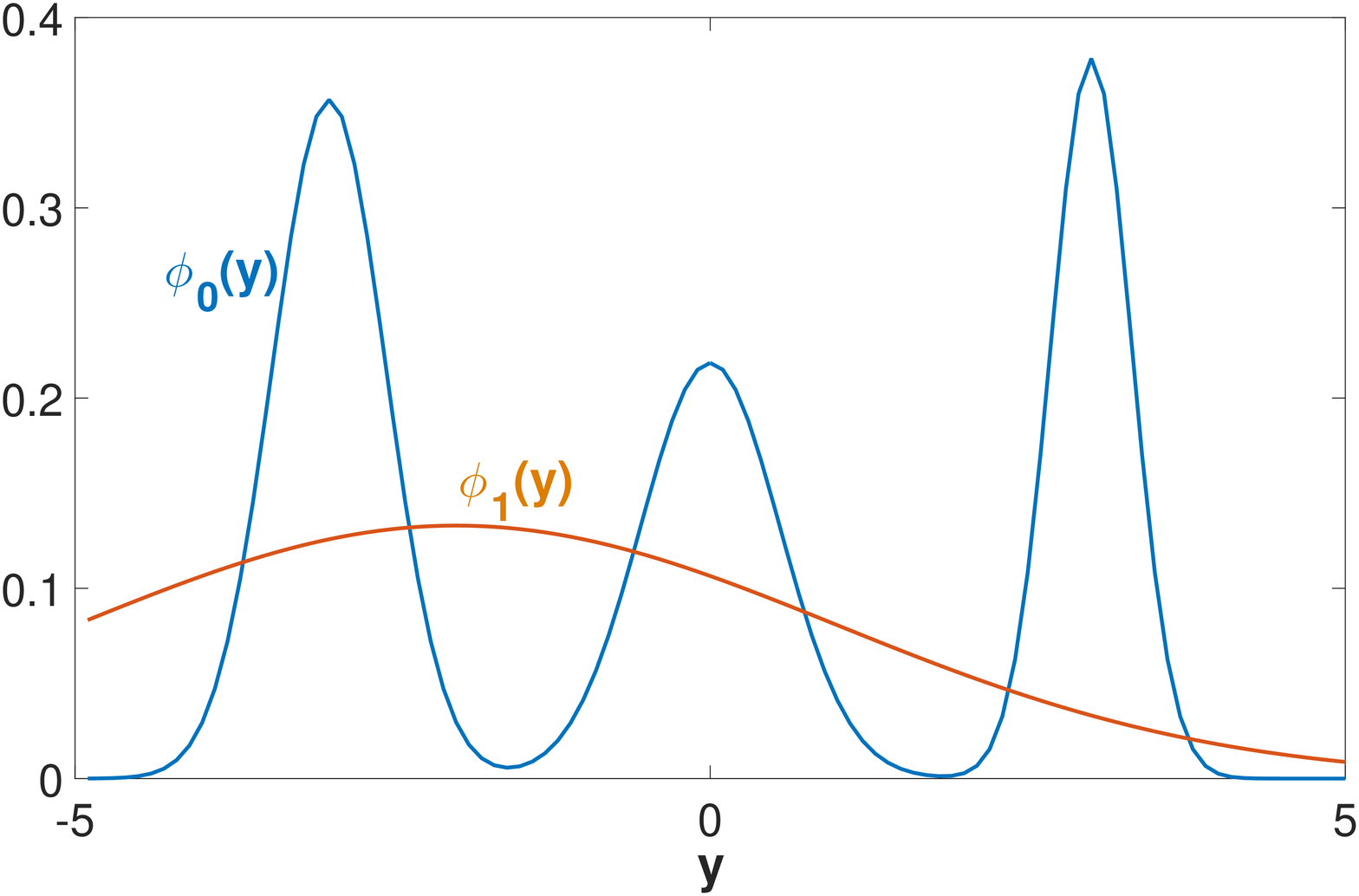}
	\caption{ Conditional densities  $\phi_0(y)=0.3 N(0,\sqrt{0.3} )+ 0.4 N(-3,\sqrt{0.2}) +0.3 N(3, \sqrt{0.1})$ and $\phi_1(y)=N(-2,3)$.  They are used in Fig. \ref{fig:cor 3} and Fig. \ref{fig:u}. }
	\label{fig: support}
\end{figure}

\begin{figure}
	\centering
	\includegraphics[width=0.5\linewidth]{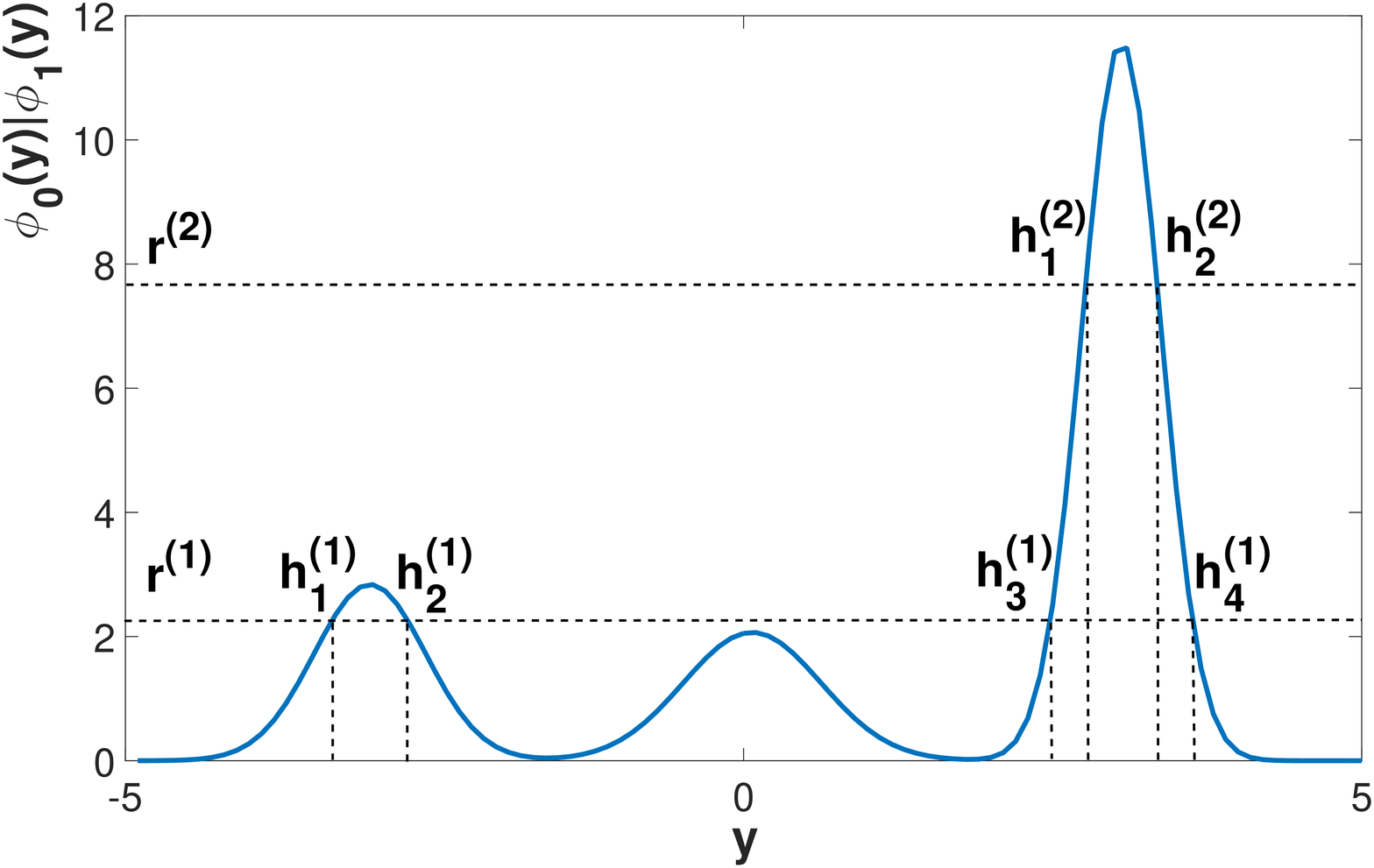}
	\caption{Two thresholding vectors $\textbf{h}^{(1)} = (h^{(1)}_1, h^{(1)}_2,h^{(1)}_3,h^{(1)}_4)$ and  $\textbf{h}^{(2)} = (h^{(2)}_1, h^{(2)}_2)$ correspond to two different values of $r$ are shown. $\phi_0(y)=0.3 N(0,\sqrt{0.3} )+ 0.4 N(-3,\sqrt{0.2}) +0.3 N(3, \sqrt{0.1})$, $\phi_1(y)=N(-2,3)$.}
	\label{fig:cor 3}
\end{figure}

\begin{figure}
	\centering
	\includegraphics[width=0.5\linewidth]{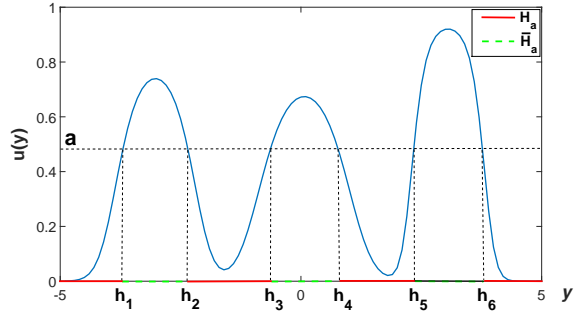}
	\caption{Illustration of the sets $\mathbb{H}_a$ and $\bar{\mathbb{H}}_a$.  $\mathbb{H}_a$ consists of solid red segments while  $\bar{\mathbb{H}}_a$ consists of green dotted segments.  In this example, there exists a quantizer with 6 thresholds $h_1, h_2, \dots, h_6$ that correspond to a specific value of $a=0.5$.  $p_0=p_1=0.5$, $\phi_0(y)=0.3 N(0,\sqrt{0.3} )+ 0.4 N(-3,\sqrt{0.2}) +0.3 N(3, \sqrt{0.1})$, $\phi_1(y)=N(-2,3)$.}
	\label{fig:u}
\end{figure}

Without the loss of generality, suppose we use the following quantizer:
\begin{equation}
z = 
\begin{cases}
0 & y \in \mathbb{H}_a, \\
1 & y \in \bar{\mathbb{H}}_a,
\end{cases}
\end{equation}
then the channel matrix of the overall DMC  is:
$$\begin{array}{cc}
A = \begin{bmatrix} f(a) & 1-f(a) \\
1-g(a)  &  g(a)
\end{bmatrix},
\end{array}$$
where $f(a) \stackrel{\triangle}{=} p(z=0|x=0)$ and $g(a) \stackrel{\triangle}{=} p(z = 1|x=1) $.  $f(a)$ and $g(a)$ can be written in terms of $\phi_0(y)$ and $\phi_1(y)$ as:

\begin{equation}
\label{eq: construct fa}
f(a)=\int_{y \in \mathbb{H}_a} \phi_0(y)dy=\int_{-\infty}^{h_1}\phi_0(y)dy + \int_{h_2}^{h_3}\phi_0(y)dy + \dots + \int_{h_{n}}^{+\infty}\phi_0(y)dy,
\end{equation}
and
\begin{equation}
\label{eq: construct ga}
g(a)=\int_{y \in \bar{\mathbb{H}}_a} \phi_1(y)dy=\int_{h_1}^{h_2}\phi_1(y)dy + \int_{h_3}^{h_4}\phi_1(y)dy + \dots + \int_{h_{n-1}}^{h_{n}}\phi_1(y)dy.
\end{equation}

As an example, if (\ref{eq: u-y}) has two solutions $(h_1,h_2)$, then the entries $f(a)$ and $g(a)$ can be constructed as: 
\begin{eqnarray}
f(a)&=&\int_{-\infty}^{h_1}\phi_0(y)dy + \int_{h_2}^{+\infty}\phi_0(y)dy\label{eq: fa for 2 noises},\\
g(a)&=&\int_{h_1}^{h_2}\phi_1(y)dy.\label{eq: ga for 2 noise}
\end{eqnarray}

Lemmas \ref{lemma: relate first derivation} and \ref{lemma: main result 4} below provide the properties of $f(a)$ and $g(a)$ and the relationship with each other.

\begin{lemma}
\label{lemma: relate first derivation}
Derivatives of $f(a)$ and $g(a)$ are related through the following equation:
 \begin{equation}
 \label{eq: relate ft gt derivation}
\frac{d f(a)}{d a}=-\dfrac{ap_0}{(1-a)p_1} \frac{d g(a)}{d a}.
 \end{equation}
\end{lemma}
\begin{proof}
Please see the proof in Appendix \ref{sec: proof lemma 1}.
\end{proof}

\begin{lemma}
\label{lemma: main result 4}  For $\forall$ $a \in (0,1)$, 

(1) $g'(a)<0$ and $f'(a)>0$.

(2) $f(a)+ g(a) \geq 1$.
\end{lemma}
\begin{proof}
Please see the proof in Appendix \ref{sec: proof lemma 2}.
\end{proof}

We are now ready to show that there is a unique stationary point $a^*$ for $I(X;Z)_a$. 

\begin{theorem}
\label{theorem: 3}
The mutual information $I(X;Z)_{a}$ has a unique stationary point $a^*$, where $$\frac{d I(X;Z)_{a}}{d a}\big|_{a = a^*} = 0.$$
\end{theorem}

\begin{proof}
Using Lemma \ref{lemma: relate first derivation}, setting derivative of $I(X;Z)_a$ to zero, we have:
\begin{eqnarray}
\frac{d I(X;Z)_a}{d a}&=& p_0f'(a)[(\log(\dfrac{f(a)}{1-f(a)})-\log(\dfrac{p_0f(a)+p_1(1-g(a))}{p_0(1-f(a))+p_1g(a)}) \nonumber\\
&+& \dfrac{p_1^2(a-1)}{p_0^2a}(\log(\dfrac{g(a)}{1-g(a)})+ \log(\dfrac{p_0f(a)+p_1(1-g(a))}{p_0(1-f(a))+p_1g(a)})] \nonumber\\
&=& p_0f'(a)F(a) = 0,\label{eq: first derivation of I(x;z) with a}
\end{eqnarray}
where
\begin{eqnarray}
F(a)&=&\log(\dfrac{f(a)}{1-f(a)}) +  \dfrac{p_1^2(a-1)}{p_0^2a} \log(\dfrac{g(a)}{1-g(a)}) \nonumber\\
&-& \dfrac{p_0^2a + p_1^2(1-a)}{p_0^2a}\log(\dfrac{p_0f(a)+p_1(1-g(a))}{p_0(1-f(a))+p_1g(a)}).\label{eq: first derivation of I(x;z) with a1}
\end{eqnarray}

Since $f'(a) > 0$  from Lemma \ref{lemma: main result 4} and $p_0 >  0$,  the stationary points must occur at $F(a)=0$.  We will show that $F(a)$ has exactly one solution, and thus $I(X;Z)_a$ has a single stationary point $a^*$.

With a bit of algebra, we can show that $F'(a)=H(a)+G(a)$ where
\begin{eqnarray}
H(a)= \dfrac{p_1^2}{p_0^2a^2}[\log(\dfrac{g(a)}{1-g(a)})+\log(\dfrac{p_0f(a)+p_1(1-g(a))}{p_0(1-f(a))+p_1g(a)})],
\end{eqnarray}
and 
\begin{eqnarray}
G(a)&=& \dfrac{f'(a)}{f(a)(1-f(a))} + \dfrac{p_1^2(a-1)}{p_0^2a} \dfrac{g'(a)}{g(a)(1-g(a))}\nonumber\\ 
&-&  \dfrac{p_0^2a+p_1^2(1-a)}{p_0^2a}\dfrac{p_0f'(a)-p_1g'(a)}{(p_0f(a)+p_1(1-g(a)))(p_0(1-f(a))+p_1g(a))} \nonumber \\
&=& \dfrac{f'(a)}{f(a)(1\!-\!f(a))} \!+\! \dfrac{p_1^2(a-1)}{p_0^2a} \dfrac{g'(a)}{g(a)(1\!-\!g(a))} \!-\! \dfrac{p_0^2a \!+\!p_1^2(1\!-\!a)}{p_0^2a}\dfrac{p_0f'(a) \!-\!p_1g'(a)}{A},
\end{eqnarray}
where $A=[p_0f(a)+p_1(1-g(a))][p_0(1-f(a))+p_1g(a)]$. 

Next, we will show that $H(a) \geq 0$ and $G(a) \geq 0$, therefore, $F'(a) \geq 0$. To show that $H(a) \geq 0$, we note that:
\begin{eqnarray}
H(a) &=& \dfrac{p_1^2}{p_0^2a^2}\Big[\log\Big(\dfrac{g(a)}{1-g(a)}\Big)+\log\Big(\dfrac{p_0f(a)+p_1(1-g(a))}{p_0(1-f(a))+p_1g(a)}\Big)\Big]\\
&=& \dfrac{p_1^2}{p_0^2a^2} \log\Big(\dfrac{g(a)(p_0f(a)+p_1(1-g(a)))}{(1-g(a))(p_0(1-f(a))+p_1g(a))}\Big) \label{eq:a}\\
& \geq & \dfrac{p_1^2}{p_0^2a^2} \log\Big(\dfrac{g(a)(p_0(1-g(a))+p_1(1-g(a)))}{(1-g(a))(p_0g(a)+p_1g(a))}\Big) 
\label{eq:b} \\
&=& \dfrac{p_1^2}{p_0^2a^2} \log\Big(\dfrac{g(a)(1-g(a))}{(1-g(a))g(a)}\Big) \label{eq:c}\\
& = & \dfrac{p_1^2}{p_0^2a^2} \log{1} \label{eq:d}\\
&=& 0,
\end{eqnarray} 
where (\ref{eq:b}) is obtained by using $f(a) + g(a) \geq 1$ in Lemma \ref{lemma: main result 4}, specifically by replacing $f(a)$ with $1-g(a)$ in the numerator and $1-f(a)$ with $g(a)$ in the denominator of (\ref{eq:a}). (\ref{eq:c}) is obtained by noting that $p_0 + p_1$ = 1, and (\ref{eq:d}) is obtained using $1-f(a) \leq g(a)$. 

%

To show that $G(a) \geq 0$, we have:
\begin{small}
\begin{eqnarray}
G(a)&=& \dfrac{f'(a)}{f(a)(1\!-\!f(a))} \!+\! \dfrac{p_1^2(a-1)}{p_0^2a} \dfrac{g'(a)}{g(a)(1\!-\!g(a))} \!-\! \dfrac{p_0^2a \!+\!p_1^2(1\!-\!a)}{p_0^2a}\dfrac{p_0f'(a) \!-\!p_1g'(a)}{A} \nonumber \\
&=& f'(a) [\dfrac{1}{f(a)(1-f(a))} + \dfrac{p_1^3(a-1)^2}{p_0^3a^2} \dfrac{1}{g(a)(1-g(a))} -  \dfrac{(p_0^2a+p_1^2(1-a))^2}{p_0^3a^2 A}] \label{eq: v01}\\
&\geq &  f'(a) [\dfrac{1}{f(a)(1-f(a))} + \dfrac{p_1^3(a-1)^2}{p_0^3a^2} \dfrac{1}{g(a)(1-g(a))} -  \dfrac{(p_0^2a+p_1^2(1-a))^2}{p_0^3a^2 B}] \label{eq: v02}\\
&=& f'(a) \dfrac{p_0p_1[p_0ag(a)(1-g(a))-p_1(1-a)f(a)(1-f(a))]^2}{f(a)(1-f(a)) g(a)(1-g(a))p_0^3a^2 B  } \label{eq: v03}\\
&\geq & 0. 
\label{eq: G(t)}
\end{eqnarray}
\end{small}
where (\ref{eq: v01}) is due to Lemma \ref{lemma: relate first derivation}, (\ref{eq: v02}) is due to $A \geq B$ and $B=p_0f(a)(1-f(a))+ p_1g(a)(1-g(a))$ (please see the detailed proof of $A \geq B$ in Appendix \ref{sec: proof theorem 4}),  (\ref{eq: v03}) is due to algebraic manipulation. 

\textbf{Uniqueness of the stationary point.} Since $H(a) \geq 0$ and $G(a) \geq 0$, we have $F'(a) \geq 0$. Consequently, $F(a)$ is a non-decreasing function. Thus, there are two possible cases: (1) $F(a)=0$ has a single solution $a^*$ or (2)  $F(a)=0$ has uncountable number of solutions.  In the case (2), $F(a)$ must align with the x-axis starting from $a^*$ to $a^*+\epsilon$, for some $a^*$ and $\epsilon >0$. In Appendix \ref{apd: unique a fa+ga=0}, we show that case (2) is impossible. Thus, the $F(a)=0$ has a unique solution which implies that the stationary point of $I(X;Z)_a$ is unique.
\end{proof}

Since $a=\dfrac{1}{1+\dfrac{p_0\phi_0(y)}{p_1\phi_1(y)}}$ is a one-to-one mapping, the single optimal value $a^*$ maps to a single optimal 
$r^* = \dfrac{\phi_0(y)}{\phi_1(y)}$ and finally maps to a unique optimal quantizer $Q^n_{r^*}$.  We note again that there might be other optimal quantizers $Q^k_{r^*}$ with $k < n$.

{\bf Proof of part (b) in Corollary \ref{cor: 1}:} In the special case where $r(y)$ is strictly monotonic function, using Corollary \ref{cor: 1} and the fact that $r^*$ is unique, we conclude that the optimal quantizer not only has one threshold but is also unique.

%

{\bf Algorithmic implication.}  We note that to find the maximum $I(X;Z)_a$, from the analysis, one can either analytically solve for $F(a^*) = 0$ directly when possible (depending on the forms of $f(a)$, $g(a)$, etc.),  or more likely numerically find $a^*$ through an algorithm when analytical form is not possible.  In the latter, since $F(a)$ is a strictly monotonic function, one can use a bisection algorithm \cite{conte2017elementary} to find the single solution of $F(a^*) = 0$ quickly.  In particular, if for some $a_1 < a_2$ and if $F(a_1) < 0$  and $F(a_2) > 0$, one can evaluate $F(\frac{a_1 + a_2}{2})$ to determine whether it is larger or smaller than  0.  If it is larger than 0, we repeat the process on the interval $[a_1, \frac{a_1 + a_2}{2}]$.  Otherwise, we repeat the process on the interval $[\frac{a_1 + a_2}{2}, a_2]$.  The process repeats until the solution is found, i.e., within some $\epsilon$ away from zero.  As seen, the algorithm is very efficient with  $O(\log{N})$ divisions where $N \sim O(1/\epsilon)$. Once we find $a^*$, we then solve for $\textbf{h}^*$ numerically.  In Section \ref{sec:simulations}, we use this approach to find the optimal quantizers.

\section{Numerical Results}
\label{sec:simulations}
\subsection{Optimality of Single Threshold Quantizer}
\begin{exmp}
\label{ex: 2}
\end{exmp}
In this section, we consider a simple example to illustrate the optimality condition for single-threshold quantizer.
We consider a channel having $p_0=0.5$, $p_1=0.5$ and $\phi_0(y)=N(\mu=-1,\sigma=1)$, $\phi_1(y)=N(\mu=1,\sigma=1)$, respectively.  Effectively, this is the channel with the same i.i.d Gaussian noises $N(\mu=0,\sigma=1)$ being added to the inputs, where the input $X=\{ x_0 = -1, x_1 = 1 \}$.  Mathematically,
$$y=x_i+n, n \sim N(0,1),  i = 0, 1. $$



 Fig. \ref{fig:r} shows $r(y)=\frac{\phi_0(y)}{\phi_1(y)}$  as a function $y$. As seen,  $r(y)$ is a decreasing function.  Therefore, based on Corollary \ref{cor: 2},  we should expect that the optimal quantizer should have only one single threshold. Indeed, based on Theorem \ref{theorem: 3}, solving $F(a^*) = 0$ we obtain $a^* = 0.5$.  Since $r(y)$ is decreasing, $a(y)$ is increasing, and therefore, there is one unique optimal $h^*_1$ for the optimal $a^* = 0.5$ as indicated by the intersection of the horizontal red line and the blue curve in Fig. \ref{fig:a}.  The mutual information $I(X;Z)$ as the function of $a$ is numerically plotted in Fig. \ref{fig:I}. As seen, the optimal point is unique at $a^*=0.5$  which corresponds to $h^*_1=0$ and $I(X;Z)_{a^*}=0.84339$ bits which confirms our Theorems \ref{theorem: 3}.  
\begin{figure}
  \centering
  \includegraphics[width=3.2 in]{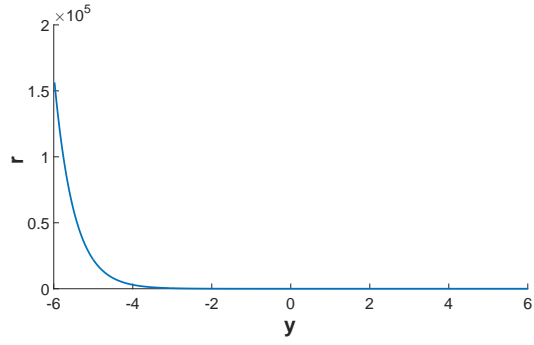}\\
  \caption{$r(y)$ as a function of $y$.}\label{fig:r}
 \end{figure}

\begin{figure}
  \centering
  \includegraphics[width=3.2 in]{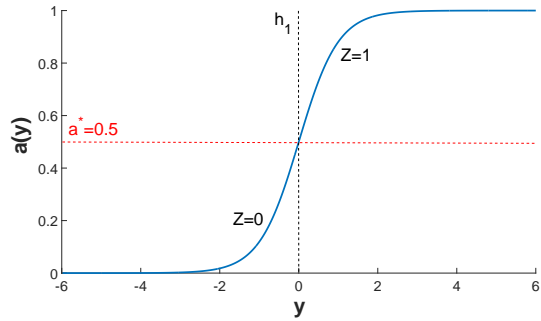}\\
  \caption{$a(y)$ as a function of $y$. $a^* = 0.5$ is the optimal point.}\label{fig:a}
 \end{figure}
 
\begin{figure}
  \centering
  \includegraphics[width=3.2 in]{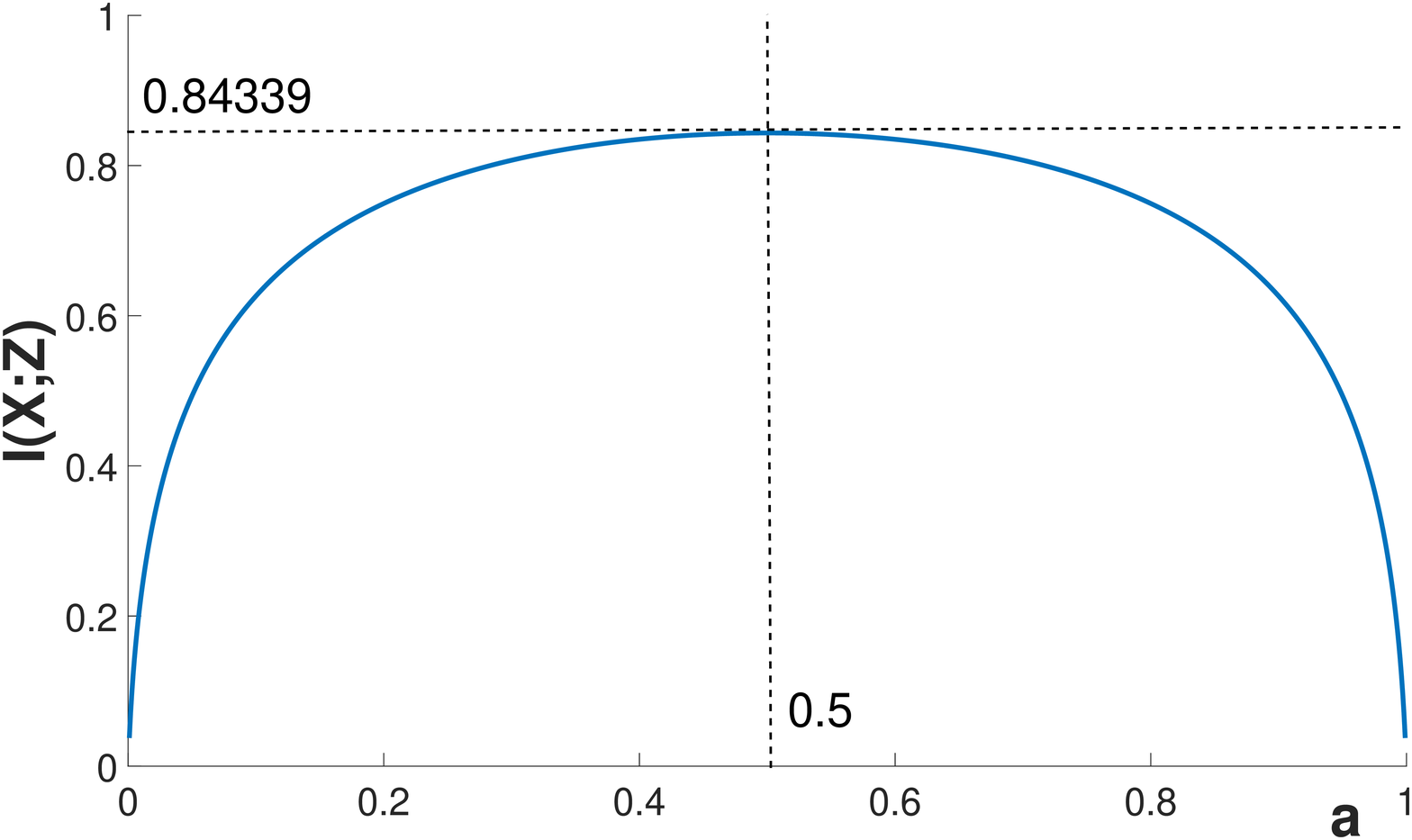}\\
  \caption{$I(X;Z)$ as a function of $a$. $a^* = 0.5$ is the optimal point that corresponds to the maximum $I(X;Z)_{a^*}=0.84339$ bits. }\label{fig:I}
 \end{figure}

\textbf{Remark:} We note that if $p_0=p_1=1/2$ and $\phi_0(y)$ and $\phi_1(y)$ are symmetry uni-modal and  i.i.d distributions, the optimal quantizer can be achieved using a single threshold $\textbf{h}^*$ which is the solution of $\phi_0(y)=\phi_1(y)$. 
For example, $\phi_0(y)=N(\mu_1,\sigma)$ and $\phi_1(y)=N(\mu_2,\sigma)$ are two normal distributions having the same variance $\sigma$, the optimal quantizer can be achieved using a single threshold $h_1^*=\dfrac{\mu_1+\mu_2}{2}$.  This is because  $h_1^*=\dfrac{\mu_1+\mu_2}{2}$ corresponding to $a^*=0.5$ and $f(a^*)=g(a^*)$ which leads to $F(a^*)=0$.  Example \ref{ex: 2} confirms our notation. 


\subsection{Unique Optimal Binary Quantizer}
\begin{exmp}
\label{ex: 3}
\end{exmp}
In this section, we show an example where the optimal quantizer requires more than one threshold.  Consider a channel having $p_0=0.5$, $p_1=0.5$ and $\phi_0(y)=N(\mu_0=-1,\sigma_0=\sqrt{5})$, $\phi_1(y)=N(\mu_1=1,\sigma_1=1)$. Effectively, this is the channel with two different Gaussian noises added to the inputs, depending whether the input $X=x_0 = -1$ or $X=x_1 = 1$.  If $X =x_0= -1$ then the noise is distributed as $N(\mu=0,\sigma=\sqrt{5})$.  If $X=x_1 = 1$ then the noise is distributed as $N(\mu=0,\sigma=1).$ Mathematically,

$$y = \begin{cases}
	x + n_0 & \text{  if  } X = x_0=-1, \\
	x + n_1 & \text{  if  } X = x_1=1,
	\end{cases}$$
	
where $n_0 \sim N(0,\sqrt{5})$ and $n_1 \sim N(0,1)$.

Fig. \ref{fig:r2} shows $r(y)=\frac{\phi_0(y)}{\phi_1(y)}$  as a function $y$. As seen,  $r(y)$ is not a monotonic function. Thus, an optimal quantizer might require multiple thresholds. Based on Theorem \ref{theorem: 3}, by solving $F(a^*) = 0$ we obtain $a^* = 0.412$.  Since $r(y)$ is not a monotonic function, $a(y)$ is also not a monotonic function as shown in Fig. \ref{fig:a2}.  As a result, there are two thresholds ($h^*_1 = -0.5374, h^*_2 = 3.5374$) correspond to the optimal $a^* = 0.412$ as indicated by the two intersections of the horizontal red line with the curve in Fig. \ref{fig:a2}. 
The mutual information $I(X;Z)$ as the function of $a$ is numerically plotted in Fig. \ref{fig:I2}. As seen, the optimal point is unique at $a^*=0.412$  which corresponds $I(X;Z)_{a^*}=0.40093$ bits.


 \begin{figure}
  \centering
  \includegraphics[width=3.2 in]{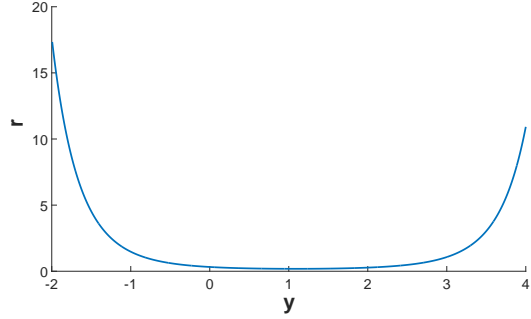}\\
  \caption{Optimal two-threshold quantizer: $r(y)$ as a non-monotonic function of $y$.}\label{fig:r2}
 \end{figure}

 \begin{figure}
  \centering
  \includegraphics[width=3.2 in]{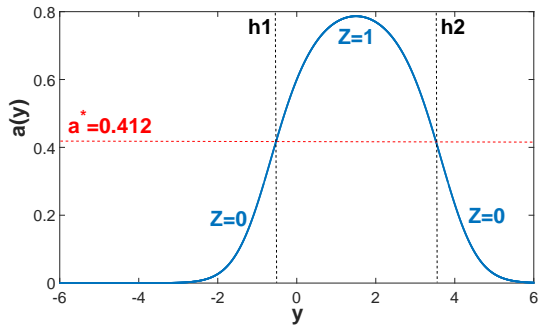}\\
  \caption{Optimal two-threshold quantizer: $a(y)$ as a non-monotonic function of $y$. There are two thresholds correspond to the optimal value $a^*=0.412$. }\label{fig:a2}
 \end{figure}
 
   \begin{figure}
  \centering
  \includegraphics[width=3.2 in]{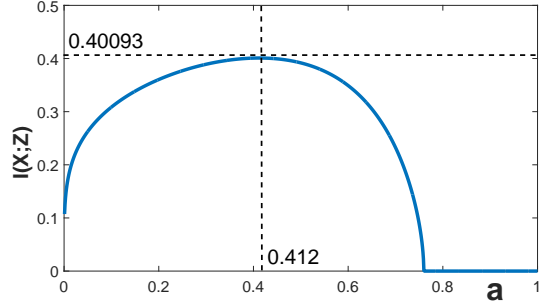}\\
  \caption{Mutual information $I(X;Z)_a$ as a function of $a$, with a unique optimal point $a^* = 0.412$.}\label{fig:I2}
 \end{figure}

\section{Conclusion}

In this paper, we show that there is a unique optimal binary quantizer with a single threshold  when the ratio of the channel conditional densities of the inputs $r(y) = \frac{P(Y=y|X=0)}{P(Y = y|X=1)}$ is a strictly increasing/decreasing function. 
Furthermore, we show that an optimal quantizer (possibly with multiple thresholds) is the one with the thresholding vector whose elements are all the solutions of $r(y)=r^*$ for some unique constant $r^*>0$. This uniqueness property of the optimal constant $r^*$ allows for fast algorithmic implementation such as a bisection algorithm to find the optimal quantizer.   We show numerical results for applying the proposed quantizer design technique for channels with additive Gaussian noises.

\appendix

\subsection{Proof for Lemma \ref{lemma: relate first derivation}}
\label{sec: proof lemma 1}

From (\ref{eq: relate phi0 phi1}), we have:
 \begin{equation}
 \label{eq: solving equation}
 \phi_1(h_i)=\dfrac{ap_0}{(1-a)p_1} \phi_0(h_i),   \forall i \in \{1,2,\dots,n \}.
 \end{equation} 
 
 From (\ref{eq: construct fa}) and (\ref{eq: construct ga})
 \begin{equation}
 \label{eq: fa derivation}
\frac{d f(a)}{d a}=\frac{\partial f(a)}{\partial h} \frac{\partial h}{\partial a}=+ \phi_0(h_1) \frac{\partial h_1}{\partial a}-\phi_0(h_2) \frac{\partial h_2}{\partial a}+ \dots - \phi_0(h_{n}) \frac{\partial h_n}{\partial a},
 \end{equation} 
  \begin{equation}
   \label{eq: ga derivation}
 \frac{d g(a)}{d a}=\frac{\partial g(a)}{\partial h} \frac{\partial h}{\partial a}=-\phi_1(h_1) \frac{\partial h_1}{\partial a}+\phi_1(h_2) \frac{\partial h_2}{\partial a}- \dots + \phi_1(h_{n})\frac{\partial h_n}{\partial a}.
 \end{equation}

Combining Eqs. (\ref{eq: solving equation}), (\ref{eq: fa derivation}) and (\ref{eq: ga derivation}), we have the desired proof. We note that $f'(a)$ and $g'(a)$ have the opposite sign.  As a result, if $f(a)$ increases, then $g(a)$ decreases and vice-versa. 
\begin{flushright} $\blacksquare$ \end{flushright}

\subsection{Proof for Lemma \ref{lemma: main result 4}}
\label{sec: proof lemma 2}

\textbf{(1)} From (\ref{eq: u-y}), $f(a)$ represents the quantized bit ``0" which is the area of $u(y)$ where $u(y) < a$.  Therefore, if $a$ is increasing, $f(a)$ is obviously increasing. Thus, $f'(a) >0$.  A similar proof can be established for $g(a)$ which corresponds to the area of $u(y)$ where $u(y) \geq a$.

\textbf{(2)}  We note that $f(a)$ and $g(a)$ represent the quantized bits ``0" and ``1" which correspond to the areas of $u(y) < a$ and $u(y) \geq a$, respectively. Let $\mathbb{H}_a=\{ y| u(y) < a \}$ and $\bar{\mathbb{H}}_a=\{ y| u(y) \geq a \}$. From (\ref{eq: u-y}) 
\begin{equation}
\label{eq: 53}
ap_0\phi_0(y) > (1-a)p_1 \phi_1(y), \forall y \in \mathbb{H}_a,
\end{equation}
\begin{equation}
\label{eq: 54}
ap_0\phi_0(y) \leq (1-a)p_1 \phi_1(y), \forall y \in \bar{\mathbb{H}}_a.
\end{equation}

We consider two possible cases: $a > p_1$ and $a \leq p_1$.  In both cases, we will show that $f(a) + g(a) \geq 1$.

$\bullet$ If $a \le p_1$ then $1-a \ge 1-p_1=p_0$. Thus, from (\ref{eq: 53}),  $\phi_0(y) \geq \phi_1(y)$ for $\forall$ $y \in \mathbb{H}_a$.
Therefore,
\begin{eqnarray}
f(a)+g(a)&=& \int_{y \in \mathbb{H}_a} \phi_0(y)dy + \int_{y \in \bar{\mathbb{H}}_a} \phi_1(y)dy\\
& \geq & \int_{y \in \mathbb{H}_a} \phi_1(y)dy + \int_{y \in \bar{\mathbb{H}}_a} \phi_1(y)dy\\
&=&1. \label{eq: 56}
\end{eqnarray}

$\bullet$ If $a > p_1$ then $1-a < 1- p_1 =p_0$. Thus, from (\ref{eq: 54}), $\phi_0(y) < \phi_1(y)$ for $\forall$ $y \in \bar{\mathbb{H}}_a$. 
Therefore,
\begin{eqnarray}
f(a)+g(a)&=& \int_{y \in \mathbb{H}_a} \phi_0(y)dy + \int_{y \in \bar{\mathbb{H}}_a} \phi_1(y)dy\\
& > & \int_{y \in \mathbb{H}_a} \phi_0(y)dy + \int_{y \in \bar{\mathbb{H}}_a} \phi_0(y)dy\\
&=&1. \label{eq: 55}
\end{eqnarray}

\begin{flushright} $\blacksquare$ \end{flushright}

\textbf{Remark:} From (\ref{eq: 56}), we  note that the inequality becomes equality if and only if $a=p_1$ and $\phi_0(y)=\phi_1(y)$ for $\forall$ $y \in {\mathbb{H}}_a$. 

\subsection{Proof for $A \geq B$}
\label{sec: proof theorem 4}

From the fact that $(f(a)+g(a)-1)^2 \geq  0$, with a bit of algebra, we have,
\begin{equation}
\label{eq: v05}
f(a)g(a) + (1-f(a))(1-g(a)) \geq f(a)(1-f(a))+ g(a)(1-g(a)).
\end{equation}

Thus, 
\begin{eqnarray}
A&=&(p_0f(a)+p_1(1-g(a)))(p_0(1-f(a))+p_1g(a)) \nonumber\\
&=&p_0^2f(a)(1-f(a)) + p_1^2g(a)(1-g(a))+ p_0p_1(f(a)g(a)+(1-f(a))(1-g(a))) \\
&\geq & p_0^2f(a)(1-f(a)) + p_1^2g(a)(1-g(a)) + p_0p_1(f(a)(1-f(a))+g(a)(1-g(a)))\label{eq: v06}\\
&=& p_0(p_0+p_1)f(a)(1-f(a))+ p_1(p_0+p_1)g(a)(1-g(a))\\
&=& p_0f(a)(1-f(a))+ p_1g(a)(1-g(a)) \label{eq: v07}\\
&=& B. \label{eq: make smaller}
\end{eqnarray}
with (\ref{eq: v06}) due to (\ref{eq: v05}) and (\ref{eq: v07}) due to $p_0+p_1=1$. Thus, $A \geq B$. \begin{flushright} $\blacksquare$ \end{flushright}

\subsection{Proof that $F(a)=0$ has only a single solution}
\label{apd: unique a fa+ga=0}

Since $F'(a) \ge 0$, $F(a)$ is a non-decreasing function. If $F(a) = 0$ has more than a single solution $a^*$, then it must have uncountable number of solutions  $a \in [a^*, a^* + \epsilon]$  for some $a^*$ and  $\epsilon > 0$. 
Now, since $F'(a)=H(a)+G(a)$ and $H(a) \geq 0$, $G(a) \geq 0$ for  $\forall$ $a$, a necessary condition for $F('a) = 0$ is that $H(a)=0$ for $\forall$ $a \in [a^*, a^* + \epsilon]$. However, from (\ref{eq:c}), $H(a)=0$ for $\forall$ $a \in [a^*, a^* + \epsilon]$ is equivalent to $f(a)+g(a)=1$ for $\forall$ for $a \in [a^*, a^* + \epsilon]$. This statement contradicts to our remark at the end of Appendix \ref{sec: proof lemma 2} such that $f(a)+g(a)=1$ if and only if $a=p_1$ and $\phi_0(y)=\phi_1(y)$ for $\forall$ $y \in {\mathbb{H}}_a$. Thus, there exists a single $a^*$ such that $F(a^*)=0$. \begin{flushright} $\blacksquare$ \end{flushright}
\bibliographystyle{unsrt}
\bibliography{sample}

\begin{thebibliography}{10}

\bibitem{lloyd1982least}
Stuart Lloyd.
\newblock Least squares quantization in pcm.
\newblock {\em IEEE transactions on information theory}, 28(2):129--137, 1982.

\bibitem{gersho2012vector}
Allen Gersho and Robert~M Gray.
\newblock {\em Vector quantization and signal compression}, volume 159.
\newblock Springer Science \& Business Media, 2012.

\bibitem{goldberg1986image}
Morris Goldberg, P~Boucher, and Seymour Shlien.
\newblock Image compression using adaptive vector quantization.
\newblock {\em IEEE Transactions on Communications}, 34(2):180--187, 1986.

\bibitem{gong2014compressing}
Yunchao Gong, Liu Liu, Ming Yang, and Lubomir Bourdev.
\newblock Compressing deep convolutional networks using vector quantization.
\newblock {\em arXiv preprint arXiv:1412.6115}, 2014.

\bibitem{akarun1997adaptive}
Lale Akarun, Y~Yardunci, and A~Enis Cetin.
\newblock Adaptive methods for dithering color images.
\newblock {\em IEEE transactions on image processing}, 6(7):950--955, 1997.

\bibitem{cover2012elements}
Thomas~M Cover and Joy~A Thomas.
\newblock {\em Elements of information theory}.
\newblock John Wiley \& Sons, 2012.

\bibitem{nguyen2018closed}
Thuan Nguyen and Thinh Nguyen.
\newblock On closed form capacities of discrete memoryless channels.
\newblock In {\em 2018 IEEE 87th Vehicular Technology Conference (VTC Spring)},
  pages 1--5. IEEE, 2018.

\bibitem{romero2015decoding}
Francisco Javier~Cuadros Romero and Brian~M Kurkoski.
\newblock Decoding ldpc codes with mutual information-maximizing lookup tables.
\newblock In {\em Information Theory (ISIT), 2015 IEEE International Symposium
  on}, pages 426--430. IEEE, 2015.

\bibitem{wang2011soft}
Jiadong Wang, Thomas Courtade, Hari Shankar, and Richard~D Wesel.
\newblock Soft information for ldpc decoding in flash: Mutual-information
  optimized quantization.
\newblock In {\em Global Telecommunications Conference (GLOBECOM 2011), 2011
  IEEE}, pages 1--6. IEEE, 2011.

\bibitem{tal2011construct}
Ido Tal and Alexander Vardy.
\newblock How to construct polar codes.
\newblock {\em arXiv preprint arXiv:1105.6164}, 2011.

\bibitem{max1960quantizing}
Joel Max.
\newblock Quantizing for minimum distortion.
\newblock {\em IRE Transactions on Information Theory}, 6(1):7--12, 1960.

\bibitem{smith1971information}
Joel~G Smith.
\newblock The information capacity of amplitude-and variance-constrained sclar
  gaussian channels.
\newblock {\em Information and Control}, 18(3):203--219, 1971.

\bibitem{kurkoski2014quantization}
Brian~M Kurkoski and Hideki Yagi.
\newblock Quantization of binary-input discrete memoryless channels.
\newblock {\em IEEE Transactions on Information Theory}, 60(8):4544--4552,
  2014.

\bibitem{mathar2013threshold}
Rudolf Mathar and Meik D{\"o}rpinghaus.
\newblock Threshold optimization for capacity-achieving discrete input one-bit
  output quantization.
\newblock In {\em Information Theory Proceedings (ISIT), 2013 IEEE
  International Symposium on}, pages 1999--2003. IEEE, 2013.

\bibitem{sakai2014suboptimal}
Yuta Sakai and Ken-ichi Iwata.
\newblock Suboptimal quantizer design for outputs of discrete memoryless
  channels with a finite-input alphabet.
\newblock In {\em Information Theory and its Applications (ISITA), 2014
  International Symposium on}, pages 120--124. IEEE, 2014.

\bibitem{iwata2014quantizer}
Ken-ichi Iwata and Shin-ya Ozawa.
\newblock Quantizer design for outputs of binary-input discrete memoryless
  channels using smawk algorithm.
\newblock In {\em Information Theory (ISIT), 2014 IEEE International Symposium
  on}, pages 191--195. IEEE, 2014.

\bibitem{winkelbauer2013channel}
Andreas Winkelbauer, Gerald Matz, and Andreas Burg.
\newblock Channel-optimized vector quantization with mutual information as
  fidelity criterion.
\newblock In {\em Signals, Systems and Computers, 2013 Asilomar Conference on},
  pages 851--855. IEEE, 2013.

\bibitem{koch2013low}
Tobias Koch and Amos Lapidoth.
\newblock At low snr, asymmetric quantizers are better.
\newblock {\em IEEE Trans. Information Theory}, 59(9):5421--5445, 2013.

\bibitem{DBLP:journals/corr/abs-1901-01659}
Xuan He, Kui Cai, Wentu Song, and Zhen Mei.
\newblock Dynamic programming for discrete memoryless channel quantization.
\newblock {\em CoRR}, abs/1901.01659, 2019.

\bibitem{burshtein1992minimum}
David Burshtein, Vincent Della~Pietra, Dimitri Kanevsky, and Arthur Nadas.
\newblock Minimum impurity partitions.
\newblock {\em The Annals of Statistics}, pages 1637--1646, 1992.

\bibitem{kurkoski2017single}
Brian~M Kurkoski and Hideki Yagi.
\newblock Single-bit quantization of binary-input, continuous-output channels.
\newblock In {\em Information Theory (ISIT), 2017 IEEE International Symposium
  on}, pages 2088--2092. IEEE, 2017.

\bibitem{coppersmith1999partitioning}
Don Coppersmith, Se~June Hong, and Jonathan~RM Hosking.
\newblock Partitioning nominal attributes in decision trees.
\newblock {\em Data Mining and Knowledge Discovery}, 3(2):197--217, 1999.

\bibitem{zhang2016low}
Jiuyang~Alan Zhang and Brian~M Kurkoski.
\newblock Low-complexity quantization of discrete memoryless channels.
\newblock In {\em 2016 International Symposium on Information Theory and Its
  Applications (ISITA)}, pages 448--452. IEEE, 2016.

\bibitem{nazer2017information}
Bobak Nazer, Or~Ordentlich, and Yury Polyanskiy.
\newblock Information-distilling quantizers.
\newblock In {\em 2017 IEEE International Symposium on Information Theory
  (ISIT)}, pages 96--100. IEEE, 2017.

\bibitem{laber2018binary}
Eduardo~S Laber, Marco Molinaro, and Felipe A~Mello Pereira.
\newblock Binary partitions with approximate minimum impurity.
\newblock In {\em International Conference on Machine Learning}, pages
  2860--2868, 2018.

\bibitem{cicalese2019new}
Ferdinando Cicalese, Eduardo Laber, and Lucas Murtinho.
\newblock New results on information theoretic clustering.
\newblock In {\em International Conference on Machine Learning}, pages
  1242--1251, 2019.

\bibitem{nguyen2018capacities}
Thuan Nguyen, Yu-Jung Chu, and Thinh Nguyen.
\newblock On the capacities of discrete memoryless thresholding channels.
\newblock In {\em 2018 IEEE 87th Vehicular Technology Conference (VTC Spring)},
  pages 1--5. IEEE, 2018.

\bibitem{alirezaei2015optimum}
Gholamreza Alirezaei and Rudolf Mathar.
\newblock Optimum one-bit quantization.
\newblock In {\em Information Theory Workshop-Fall (ITW), 2015 IEEE}, pages
  357--361. IEEE, 2015.

\bibitem{singh2009limits}
Jaspreet Singh, Onkar Dabeer, and Upamanyu Madhow.
\newblock On the limits of communication with low-precision analog-to-digital
  conversion at the receiver.
\newblock {\em IEEE Transactions on Communications}, 57(12):3629--3639, 2009.

\bibitem{mumey2003optimal}
Brendan Mumey and Tom{\'a}{\v{s}} Gedeon.
\newblock Optimal mutual information quantization is np-complete.
\newblock In {\em Neural Information Coding (NIC) workshop poster, Snowbird
  UT}, pages 1932--4553, 2003.

\bibitem{boyd2004convex}
Stephen Boyd and Lieven Vandenberghe.
\newblock {\em Convex optimization}.
\newblock Cambridge university press, 2004.

\bibitem{conte2017elementary}
Samuel~Daniel Conte and Carl De~Boor.
\newblock {\em Elementary numerical analysis: an algorithmic approach},
  volume~78.
\newblock SIAM, 2017.

\end{thebibliography}

\end{document}